\documentclass[12pt,a4paper]{amsart}
\usepackage[margin=3cm]{geometry}
\usepackage{amsthm}
\usepackage{graphicx}

\usepackage{enumitem}

\newtheorem{thm}{Theorem}[section]
\newtheorem{prop}[thm]{Proposition}

\newtheorem{lem}[thm]{Lemma}
\newtheorem{cor}[thm]{Corollary}

\theoremstyle{definition} 

\newtheorem{defn}[thm]{Definition}
\newtheorem{rem}[thm]{Remark}

\newcommand{\N}{\mathcal N}
\newcommand{\G}{\mathcal G}
\newcommand{\Ss}{\mathcal S}

\newcommand{\NNI}{N\!N\!I}
\newcommand{\SPR}{S\!P\!R}
\newcommand{\TBR}{T\!B\!R}

\title{Bounds for phylogenetic network space metrics}
\author[Francis, Huber, Moulton, Wu]{Andrew Francis, Katharina T.\,Huber, Vincent Moulton, Taoyang Wu}
\address{Francis: Centre for Research in Mathematics, Western Sydney University, Australia}
\address{Huber, Moulton, Wu: School of Computing Sciences, University of East Anglia, UK}
\date\today

\begin{document}
\begin{abstract}
Phylogenetic networks are a generalization of phylogenetic trees that allow for
representation of reticulate evolution. Recently, a space of unrooted
phylogenetic networks was introduced,
where such a network is a connected graph in which every vertex has degree 1 or 3 
and whose leaf-set is a fixed set $X$ of taxa. This space, denoted $\mathcal{N}(X)$,
is defined in terms of two operations on networks -- the nearest 
neighbor interchange and triangle operations -- which can be used to 
transform any network with leaf set $X$ into any other network 
with that leaf set. In particular, it 
gives rise to a metric $d$ on $\mathcal N(X)$ which is given by the 
smallest number of operations required to transform one network 
in $\mathcal N(X)$ into another in $\mathcal N(X)$.
The  metric generalizes the well-known NNI-metric on phylogenetic trees
which has been intensively studied in the literature. In this paper, we
derive a bound for the metric $d$ as well as a related 
metric $d_{\NNI}$
which arises when restricting $d$ to the subset of $\mathcal{N}(X)$
consisting of all networks with $2(|X|-1+i)$ vertices, $i \ge 1$.
We also introduce two new metrics  on networks -- the SPR and TBR metrics -- 
which generalize the metrics on phylogenetic trees with the same name 
and give bounds for these new metrics. We expect our results to
eventually have applications to the development and understanding 
of network search algorithms.
\end{abstract}

\maketitle

\section{Introduction}\label{s:intro}

Phylogenetic networks are a generalization of phylogenetic trees that are used 
to represent either non-tree-like evolutionary histories arising in organisms 
such as plants and bacteria, or uncertainty in evolutionary histories \cite{huson2010}. 
Here we are interested in {\em unrooted} binary phylogenetic  networks 
on a finite set $X$ of taxa, or {\em networks} for short. These 
are connected graphs in which every vertex has degree 1 or 3 
and whose leaf-set is $X$ \cite{gambette2012quartets}. 
An example of such a network is presented in Figure~\ref{f:network.nni.examples}(i).
Note that if a 
network is a tree (i.e. it has no cycles), then it is also known as a {\em phylogenetic tree}. 
Networks can be generated from biological data using software such 
as T-REX \cite{makarenkov2001trex} and have been used, for example, to study 
the origin of genomes in eukaryotes \cite{rivera2004ring}. 

Recently, it has been shown that it is possible to transform any network on a set $X$ 
into any other network on the same set using a finite sequence 
of two types of operations \cite{huber2016transforming}. 
These operations are pictured in Figure~\ref{f:network.nni.examples}(ii) and (iii),
and are called {\em nearest neighbor interchange} (NNI) and {\em triangle} operations, 
respectively.  Note that the NNI operation generalizes the operation with the same
name which is used to compare phylogenetic trees \cite{robinson1971}.
In light of this result, as observed in  \cite{huber2016transforming},
a space ${\N}(X)$ of phylogenetic networks on $X$ may be defined as follows. 
It is the graph with vertex set consisting of all networks on $X$, and 
edges corresponding to pairs of networks which differ by either one NNI operation or one triangle operation. 
Since we can transform any network in $\mathcal N(X)$ into any other network in $\mathcal N(X)$
using a finite sequence of NNI and triangle operations, it follows that the space $\mathcal N(X)$ is connected. 

The space ${\N}(X)$ generalizes tree-space \cite{billera2001geometry}, the graph with vertex set consisting of all 
phylogenetic trees on $X$ with edges corresponding to pairs of trees which differ by one NNI operation. 
Indeed, it actually contains tree-space (on $X$) as a subspace as we shall now explain. 
For $i \ge 0$, we let $\N_i(X)$ denote the set of all networks on $X$ with $2(|X|-1+i)$ 
vertices, which we call the {\em $i$'th tier}.  A tier 3 example is shown in Figure~\ref{f:network.nni.examples}(i).  Clearly the space $\N(X)$ is the disjoint union  
of the set of tiers $ \N_i(X)$ taken over $i \ge 0$. Moreover, tier $ {\mathcal N}_0(X)$ 
is precisely the set of phylogenetic trees on $X$. 
Each tier $\N_i(X)$ is a 
connected subgraph of $\N(X)$, where the edges correspond only to NNI operations~\cite{huber2016transforming}, so that tree-space is a subspace of $\N(X)$.

Tree-space is equipped with the {\em NNI metric} $d_{\NNI}$, which 
for any two trees $T$ and $T'$ contained in it 
is defined to be the 
minimum number of NNI operations required to transform $T$ into $T'$. 
The NNI metric has been intensively studied in the literature (see e.g.~\cite{dasgupta1997distances}), and
its properties have important consequences for tree search algorithms. 
One such property is the diameter of tree space, where the {\em diameter} $\Delta(D)$ of a metric 
$D$ on a set $Y$ is its maximum value taken over all pairs of elements in $Y$. 
In \cite{li1996nearest} it is shown that, for $\ell=|X|\geq 3$, the diameter of $d_{\NNI}$ satisfies
$$
(\ell-4)/2 \log[(2\sqrt{2}/3e)(\ell-2)] \le \Delta(d_{\NNI}) \le \ell \log(\ell) +O(\ell). 
$$
The second bound improved on an $O(\ell^2)$ upper bound given 
by Robinson in \cite{robinson1971}.

Network spaces are equipped with metrics which naturally generalize 
the NNI-metric on trees. In particular, for $N,N' \in \N_i(X)$ (or, more generally, $N,N' \in \N(X)$), 
we define the distance $d_{\NNI}(N,N')$ (the distance $d_{\N(X)}(N,N')$) to be the minimal number 
of NNI operations (respectively, NNI and triangle operations) to transform $N$ into $N'$. 
In this paper, we focus on giving bounds on the diameter of $d_{\NNI}$ of tier $\N_i(X)$,
and upper bounds for $d_{\N(X)}(N,N')$ for any $N,N' \in \N(X)$. 
Note that $d_{\NNI}$ is bounded on ${\N}_i(X)$ (since $|\N_i(X)|$ is finite), 
whereas $d_{\N(X)}$ can become arbitrarily large on ${\N}(X)$. Hence 
it only makes sense to consider diameter bounds for the metric $d_{\NNI}$ on ${\N}_i(X)$. 
As with tree-space, we expect that our results could eventually prove useful for network construction algorithms.

We now summarize the contents of this paper. After presenting some preliminaries 
in the next section, in Section~\ref{s:echidna} we begin by introducing 
a family of phylogenetic networks that we call ``echidna'' networks. We then exploit properties of 
these networks in Section~\ref{s:NNI.lower.bound}, together with results on 
graph grammars~\cite{sleator1992short}, to give a lower bound 
for the diameter of the metric $d_{\NNI}$ on ${\N}_i(X)$ (see Theorem~\ref{t:NNI.lower.bound}). An upper bound for 
the same diameter is then derived in Section~\ref{s:upper.bound.NNI} (see Theorem~\ref{t:NNI.diam.upper.bd}).  
To derive this bound, we exploit properties of Hamiltonian paths 
in the graph that arises from a network by removing its leaves and their adjacent edges.
Using our upper bound on $d_{\NNI}$, we also derive an upper bound for
$d_{\N(X)}(N,N')$ for any $N,N' \in \N(X)$ (see Corollary~\ref{c:NNI.dist.btw.arb.netwks}).  

In Section~\ref{s:SPR.TBR}, we define SPR and TBR operations on networks. 
These operations generalize the NNI operation, as well as the well-known subtree prune and regraft
(SPR) and tree bisection and reconnection (TBR) operations on trees (cf. \cite{allen2001subtree}). 
The SPR and TBR operations allow parts of a network to be 
chopped off and reconnected somewhere onto the resulting network,
in contrast to the NNI and triangle operations which are local in nature. 
In Section~\ref{s:SPR.TBR.bounds}, we derive 
bounds for the diameter of the SPR and TBR metrics on the set on ${\N}_i(X)$. 
We conclude in Section~\ref{s:discussion} with a discussion of some possible future directions.

\section{Preliminaries}\label{s:prelims}

For us, graphs contain no parallel edges (edges with the same pair of
end vertices), and no loops (edges with one vertex as both end vertices).
This means that edges are uniquely determined by a pair of vertices $\{v,w\}$ with $v\neq w$. 

Suppose throughout that $X$ is a finite set with $|X|\geq 3$. 
A \emph{phylogenetic network} on leaf-set $X$ (or a 
network (on $X$), for short) is a connected graph in which every vertex has degree 3 or degree 1, and in which the vertices of degree 1 are labelled by the elements of $X$ (e.g. Figure~\ref{f:network.nni.examples}(i)).  This means that a phylogenetic network is essentially a cubic graph (a graph in which every vertex has degree 3) with leaves attached.  It also means that phylogenetic networks for us are \emph{unrooted}, so that edges have no implicit direction.   

\begin{figure}[ht]
\includegraphics[width=15cm]{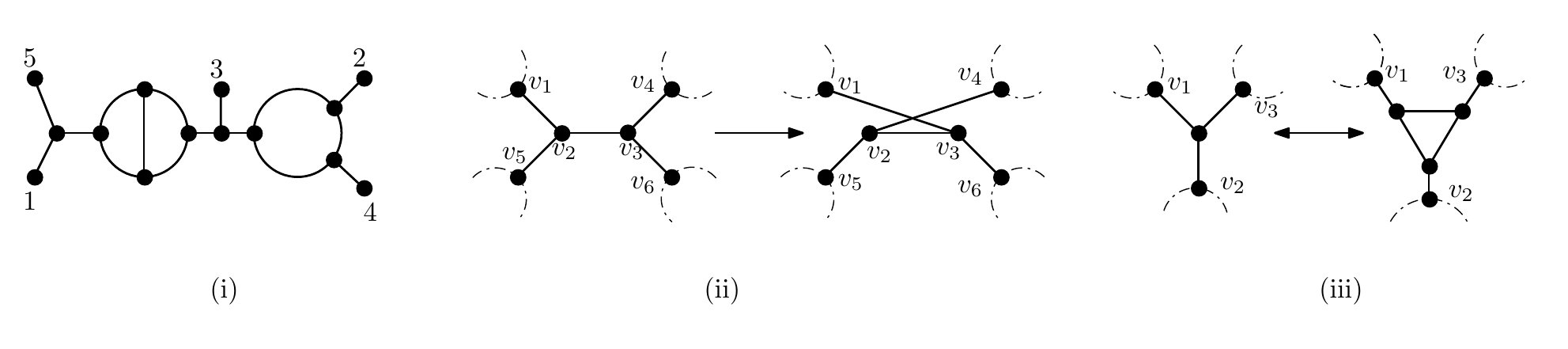}
\caption{(i) Example of a phylogenetic network on the set
	$X=\{1,2,3,4,5\}$.  This network is in tier 3, because it has $n=14$ vertices and $\ell=5$ leaves, and $14=2(5+3-1)$.  It has two blobs.  (ii) An NNI operation on adjacent degree three vertices, changing a path $v_1, v_2,v_3,v_4$ to $v_1,v_3,v_2,v_4$.  (iii) The triangle operation that shifts between tiers $\N_i(x)$ and $\N_{i+1}(X)$, $i\geq 0$.}\label{f:network.nni.examples}
\end{figure}

Write $V(N)$ for the set of vertices in $N$ and $E(N)$ for the set of edges of $N$.  We will reserve $n$ for the number of vertices in the network, $n:=|V(N)|$.

The concept of the \emph{tier} of a network on $X$ will be important for this paper, and has been defined in Section~\ref{s:intro} (following~\cite{huber2016transforming}).  It is also known as the \emph{reticulation number} of a network, because it is the number of edges one must remove from it for it to become a phylogenetic tree on $X$ (Lemma~\ref{retnum}).  
 
A \emph{cut-edge}, or \emph{bridge}, of a network is an edge whose removal disconnects the graph.  
A cut-edge is \emph{trivial} if one of the connected components induced by the cut-edge is a vertex
and {\em non-trivial} otherwise.
A \emph{simple} network is one whose cut-edges are all trivial (so note, for instance, that trees on more than two leaves are \emph{not} simple networks).
A \emph{blob} in a network is a maximal subgraph that has no cut-edge, and that is not a vertex~\cite{gambette2012quartets}.

There are several numbers associated with a network that will be widely used in this paper.  The first, $n$, has already been mentioned: $n=|V(N)|$.  Others are the size of the leaf-set, $\ell:=|X|$, and the tier of the network, which we will usually denote $i$.  These three variables are related by the following equation, as stated in the Introduction:
\[
n=2(\ell+i-1).
\]

In this paper we will consider networks that we call \emph{pseudo-Hamiltonian}:  networks that contain a cycle that passes through every non-leaf vertex.
Note that every pseudo-Hamiltonian network is simple, but not vice versa.  One can construct simple graphs that are not pseudo-Hamiltonian, by for instance taking a cubic graph that is not Hamiltonian, and adding some leaves to it.

The nearest-neighbour interchange (NNI) is a local operation, initially defined for phylogenetic trees, that is important for moving around tree-space in search algorithms.  Such algorithms are vital for estimating phylogenetic trees using likelihood or parsimony methods.  The NNI operation has also been defined as follows for phylogenetic networks~\cite{huber2016transforming}, since it is in a sense an operation on a pair of adjacent degree 3 vertices in a graph (see Figure~\ref{f:network.nni.examples}(ii)).

\begin{defn}[NNI]
Let $v_1,v_2,v_3,v_4$ be a path in a network in which neither $\{v_1,v_3\}$ nor $\{v_2,v_4\}$ is an edge.  An NNI operation on this path replaces it with the path $v_1,v_3,v_2,v_4$. 
\end{defn}

This replacement of a path has the effect of retaining the central edge $\{v_2,v_3\}$, while replacing edge $\{v_1,v_2\}$ with the new edge $\{v_1,v_3\}$ and edge $\{v_3,v_4\}$ with the new edge $\{v_2,v_4\}$.  

We now briefly digress beyond a fixed tier and consider the wider network space $\N(X)$. 
The \emph{triangle operation} introduced in~\cite{huber2016transforming}, allows movement between tiers by inserting a 3-cycle at any degree-3 vertex (``blow-up''), or collapsing a 3-cycle into a degree-3 vertex (``collapse'').  See Figure~\ref{f:network.nni.examples}(iii).

\begin{prop}[\cite{huber2016transforming}]\label{p:NX.connected}
The space of networks $\N(X)$ is connected by NNI operations together with triangle operations.
\end{prop}

Because the space $\N(X)$ is connected, the distance  $d_{\N(X)}$ is well-defined, and indeed is a metric (as is 
the NNI distance on tier $i$ networks)~\cite[Theorem 5]{huber2016transforming}. As it turns out, a canonical
extension of the notion of the subtree prune and regraft (SPR) and tree-bisect and regraft (TBR) operations for
trees to networks (see Definitions~\ref{d:SPR} and~\ref{d:TBR}
for precise details)
allows us to establish the companion result for
Proposition~\ref{p:NX.connected}.

\begin{cor}\label{c:NX.connected.SPR.TBR}
The space of networks $\N(X)$ is connected by SPR operations together with triangle operations, and by TBR operations together with triangle operations.
\end{cor}

\begin{proof}
Each SPR or TBR operation can be performed by a series of NNI operations (this is easy to check and is noted in Lemma~\ref{l:NNI.ss.SPR.ss.TBR}), so the result follows immediately from Proposition~\ref{p:NX.connected}.
\end{proof}

Finally for this preliminary section, we reiterate that this paper is focussed on movements within a single tier.  The remarks about wider movement around the space $\N(X)$ in Proposition~\ref{p:NX.connected} and Corollary~\ref{c:NX.connected.SPR.TBR} are included for context.

Write $S_k$ for the symmetric group on the set $\{1,\dots,k\}$, for $k\ge 1$.  For the sake of extremal cases, we set $S_0$ to be trivial group consisting of the empty map from $\emptyset$ to itself.  Similarly, we adopt the convention that $0!=1$.

\section{Echidna graphs}\label{s:echidna}

The first main result of this paper, provided in Section~\ref{s:NNI.lower.bound}, is a lower bound on the diameter of the space $\N_i(X)$ of tier $i$ phylogenetic networks under NNI operations.  To obtain this bound, we will need a lower bound on the number of phylogenetic networks in tier $i$.  That lower bound is established in this section (Corollary~\ref{c:number.of.networks}), by counting the number of distinct networks in a subset of $\N_i(X)$.  This subset is the set of \emph{echidna} graphs, which we will define shortly.  Echidna graphs are useful for this purpose because they can be counted through a bijection with a set of sequences $\Ss(p,q)$, defined as follows.

For integers $p\ge 1$ and $q\ge 0$, define $\Ss(p,q)$ to be the set of sequences of length $p+q$ whose entries are the symbols $\{a_1,\dots,a_p\}$ and $q$ copies of 0, and that begin and end with $a_1$ and $a_p$ respectively.  Denote the $k$-th entry of a sequence $S\in \Ss(p,q)$ by $S[k]$.  The number of such sequences is $|\Ss(p,q)|=\frac{(p+q-2)!}{q!}$.

For $\ell\ge 3$ and $i\ge 1$, we use a sequence $S\in \Ss(\ell,i-1)$ and a permutation $\pi\in S_{i-1}$ (if $i>1$), to define a tier $i$ network $G(S,\pi)$ with $\ell$ leaves labelled $\{1,\dots,\ell\}$ as follows. 

Draw a circle and create $\ell+i-1$ degree 2 vertices labelled clockwise by the sequence elements $S[k]$, for 
$1\leq k\leq \ell+i-1$, to obtain a cycle $C$ with vertices $S[1],\dots,S[\ell+i-1]$. 
Each vertex is thus labelled $a_j$ for some $j=1,\dots,\ell$ or 0. 
To each of the $\ell$ vertices for which $S[k]\neq 0$, attach a leaf with label $k$. 
Next, subdivide the edge $\{S[1], S[\ell+i-1]\}$ by $i-1$ degree 2 vertices 
reading anticlockwise from $S[1]$ to $S[\ell+i-1]$.  Referring to the resulting graph also as $C$, draw $i-1$ {\em chords}, that is, edges from the degree 2 vertices along the top of $C$ (those labelled 0) to the degree 2 vertices along the bottom of $C$ according to the permutation $\pi$ (using implied numbering from their positions in the sequence).  Denote this graph $G(S,\pi)$. An example with $\pi=id$ is shown in Figure~\ref{f:network.family.identity}.

\begin{figure}[ht]
\begin{center}
\includegraphics[width=0.3\textwidth]{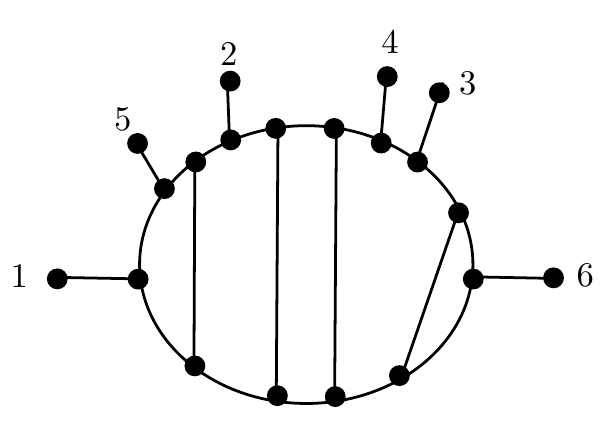}
\end{center}
\caption{Example of a phylogenetic network $G(S,\pi_0)$ in the echidna family with $S=(a_1,a_5,0,a_2,0,0,a_4,a_3,0,a_6)$.}
\label{f:network.family.identity}
\end{figure}

We call graphs constructed in this way \emph{echidna graphs}, and for a given number of leaves $\ell$ and $i\ge 1$, denote the set of such graphs 
\[
\G(\ell,i-1):=\{G(S,\pi)\mid S\in \Ss(\ell,i-1),\pi\in S_{i-1}\}.
\]  
Note that elements of $\G(\ell,i-1)$ are tier $i$ phylogenetic networks, and are also \emph{pseudo-Hamiltonian graphs}, as defined in Section~\ref{s:prelims}.

In what follows, we will restrict our attention to echidna graphs in which the permutation in
$S_{i-1}$ is the identity map $\pi_0$.
We will show that different sequences in $\Ss(\ell,i-1)$ generate non-isomorphic graphs, and begin by noting some properties of echidna graphs generated from different sequences.

Lemma~\ref{l:echidna.equal.subgraphs} states that if graphs are obtained from two sequences that are the same up to the $k$-th entry, then distances from the leaf labelled 1 to any of the vertices labelled by the first $k-1$ sequence entries are the same in both graphs.  This follows from the construction, using the identity permutation $\pi_0$.  The lemma also formalises the point that if the last point of agreement in the two sequences is a zero, then paths from leaf 1 to beyond that point must go through one of two vertices (the one that's the last point of agreement, and the other end of the chord it connects to), and focuses this property on a particular vertex in the $k$-th position of one of the sequences ($\alpha$) that we will use for later argument.

\begin{lem}\label{l:echidna.equal.subgraphs}
Write $G=G(S,\pi_0)$ and $G'=G(S',\pi_0)$ in $\G(\ell,i-1)$, for $\ell\ge 3$. 

Suppose $S\neq S'$ and let $k$ be the first position at which they differ, that is, $k\in\{1,\ldots,\ell+i-1\}$ is such that 
$S[k]\neq S'[k]$ and $S[j]=S'[j]$, for all $1\leq j\leq k-1$. 
Suppose, without loss of generality, that $S[k]\neq 0$ (noting that at least one of $S[k]$ and $S'[k]$ must be non-zero), so that $S[k]=a_\alpha$ for some $\alpha=2,\dots,\ell-1$.  Then, 
\begin{enumerate}
	\item[(i)] $d_G(1,S[j])=d_{G'}(1,S'[j])$ for $1\le j<k$. \label{l:d.equal.pre.i}
	\item[(ii)] If $S[k-1]=S'[k-1]=0$, denote the vertex in $G$  
	adjacent with $S[k-1]$, but not labelled $S[k-2]$ or $S[k]$, by $x$. Then any path from leaf $1$ to leaf $\alpha$ in $G$ must go through at least one of $S[k-1]$ or $x$.
	\item[(ii')] If $S[k-1]=S'[k-1]=0$, denote the vertex in $G'$ 
	adjacent with $S'[k-1]$, but not labelled $S'[k-2]$ or $S'[k]$, by $x$. Then any path from leaf $1$ to leaf $\alpha$ in $G'$ must go through at least one of $S'[k-1]$ or $x$.
	 \label{l:path.1.to.i}
\end{enumerate}
\end{lem}

\begin{proof}
Clear from the construction of $G$ and $G'$.
\end{proof}

Now we consider the same set-up, but with the assumption that the distances between two specific leaves in the two graphs are equal.  The two leaves are those labelled 1, and labelled $\alpha$ (the leaf corresponding to the first point that the sequences differ).

\begin{lem}\label{l:preceded.by.leaf}
	Continuing with the notation introduced in Lemma~\ref{l:echidna.equal.subgraphs}, assume that 
	$d_G(1,\alpha)= d_{G'}(1,\alpha)$ where $\alpha$ is the leaf attached to $S[k]=a_\alpha$ in $G$.
	Then either 
\begin{itemize}[itemsep=3pt]
	\item [(A)] $S[k-1]=S'[k-1]=a_\beta$ for some $\beta=1,\dots,\ell$.  
	That is, the entry before $a_\alpha$ in $S$ also corresponds to a leaf (namely $\beta$); or 

	\item [(B)] $S[k-1]=S'[k-1]=0$	is the last zero entry in $S$ (and therefore $S'$), and $a_\alpha$ is the second last entry in $S'$. That is, $S'[\ell+i-2]=S[k]$. 
\end{itemize}
\end{lem}

\begin{proof}
Firstly, we rule out the case that $a_\alpha$ is the first entry in $S$ after $a_1$, namely the case $k=2$.  If $S[2]=a_\alpha$ then $d_G(1,\alpha)=3$, and so $d_{G'}(1,\alpha)=3$ by the assumption of the Lemma.  But by construction of the echidna graphs, the only way two leaves can be 3 apart is if their corresponding terms are adjacent in the graph's defining sequence, and this forces $a_\alpha$ to also be the second entry of $S'$, a contradiction (since $S$ and $S'$ differ at the $k$-th position).

Now suppose $k>2$ and consider minimal paths from 1 to $\alpha$.  Suppose, by way of contradiction to (A), that $S[k-1]$ and $S'[k-1]=0$ (the sequences agree before the $k$-th position).  The vertex $S[k-1]$ has degree 3, with two of its neighbouring vertices being $S[k-2]$ and $S[k]=a_\alpha$ (noting $k>2$), and the third, $x$, being a vertex at the bottom of the graph connected by a chord, $c$.

By Lemma~\ref{l:echidna.equal.subgraphs}(ii), a minimal path from 1 to $\alpha$ must go through $S[k-1]$ or $x$, and so either
\begin{align*}
d_G(1,\alpha)&=d_G(1,S[k-1])+d_G(S[k-1],\alpha)\\
			&=d_G(1,S[k-1])+2,
\end{align*}
since the distance from $S[k-1]$ to $\alpha$ is 2, or
\begin{align*}
d_G(1,\alpha)&=d_G(1,x)+d_G(x,\alpha)\\
			&=d_G(1,x)+3,
\end{align*}
since the distance from $x$ to $\alpha$ is 3. This can be seen because there cetainly \emph{is} a path of length 3 from $x$ to $\alpha$ (up chord $c$ to $S[k-1]$, then to $S[k]=a_\alpha$, and then to leaf $\alpha$), and in general any path from $x$ to $\alpha$ must go up a chord, must go from $S[k\pm 1]$ to $S[k]$, and must go from $S[k]$ to $\alpha$: at least 3 steps.  

Since $d_{G}(1,\alpha)=d_{G'}(1,\alpha)$, $d_G(1,S[k-1])=d_{G'}(1,S'[k-1])$, and $d_G(1,x)=d_{G'}(1,x)$ (by Lemma~\ref{l:echidna.equal.subgraphs}(ii)), we have that in $G'$ either $d_{G'}(S'[k-1],\alpha)=2$ or $d_{G'}(x,\alpha)=3$.  The only way $S'[k-1]$ could be distance 2 from $\alpha$ in $G'$ is if $a_\alpha=S'[k]$ or $S'[k-2]$, both of which are ruled out by assumptions, so therefore $d_{G'}(x,\alpha)=3$.

Any path of length 3 from $x$ to $\alpha$ in $G'$ that goes up the chord $c$ would similarly force $a_\alpha=S'[k]$ or $S'[k-2]$, both not possible.  So the path of length 3 from $x$ to $\alpha$ in $G'$ does not go up $c$.  It also cannot go towards the preceding chord since that is further from $\alpha$.  Therefore it goes towards the leaf labelled $\ell$, from $x$.  
If there was another chord in $G'$ coming after $c$, then any path from $x$ to $\alpha$ going up that chord would have distance at least 4: the path along the bottom from $x$ to the new chord; the chord; the path along the top from the top of the chord to $a_\alpha$; and the edge to the leaf $\alpha$ itself.  This is a contradiction.  

If there is no chord coming after $c$ in $G'$, then we are in the situation of (B): $c$ is the last chord, in the $(k-1)$-th position (so that $S'[k-1]$ is the last zero in $S$ and $S'$), and the position of $\alpha$ in $G'$ must be adjacent to the final leaf, $\ell$.  This situation is illustrated in Figure~\ref{f:caseB}.
\begin{figure}[ht]
\includegraphics[width=15cm]{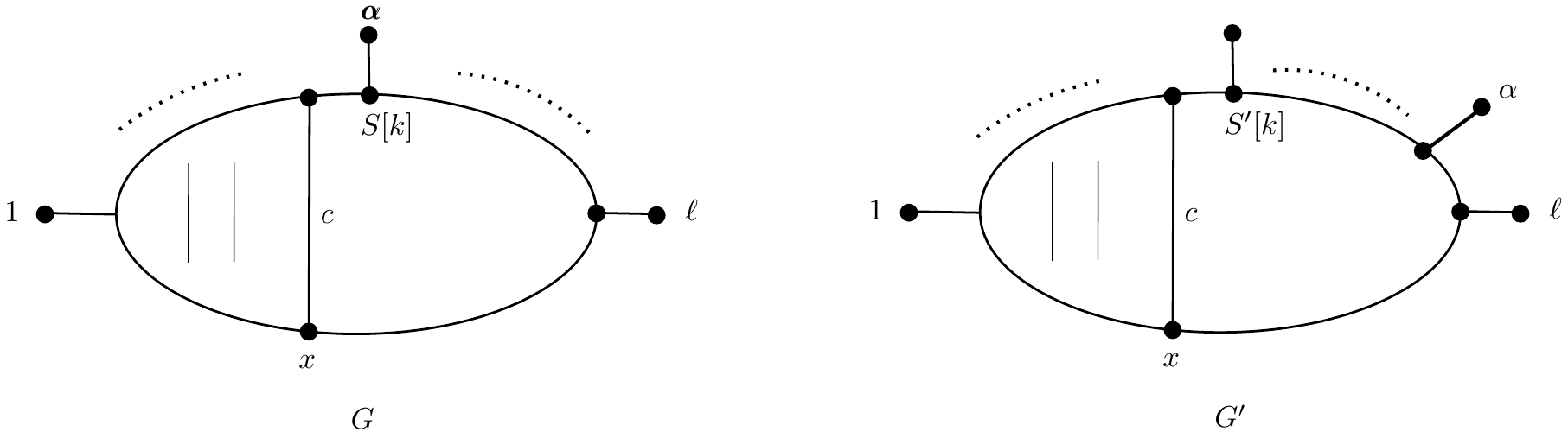}
\caption{The situation of case B in Lemma~\ref{l:preceded.by.leaf}.  All chords are to the left of chord $c$ in both graphs.  }\label{f:caseB}
\end{figure}
\end{proof}

We are now able to prove our main result about echidna graphs: that distinct (non-isomorphic) echidna graphs are generated by distinct sequences.

\begin{prop}\label{p:S.diff.implies.G.diff}
Fix $\pi_0=id\in S_{i-1}$ and write $G=G(S,\pi_0)$ and $G'=G(S',\pi_0)$.  

Let $\ell\ge 3,q\ge 0$.  If $S\neq S'$ in 
$\mathcal S(\ell, q)$, then $G(S)\not\cong G(S')$.
\end{prop}

\begin{proof}
Suppose $k$ is the first position for which $S[k]\neq S'[k]$.  Then this position is non-zero in at least one of $S$ and $S'$, and so without loss of generality suppose that $S[k]=a_\alpha$ with $\alpha>1$.  

If $d_G(1,\alpha)\neq d_{G'}(1,\alpha)$, then $G\not\cong G'$, and we are done.  So suppose that $d_G(1,\alpha)= d_{G'}(1,\alpha)$.

By Lemma~\ref{l:preceded.by.leaf}, either (A) or (B) holds.  

If (A), then there exists some leaf $\beta\neq 1$ such that $S[k-1]=S'[k-1]=a_\beta$.  But then $d_G(\alpha,\beta)=3$, while $d_{G'}(\alpha,\beta)>3$, since distances between leaves can only be 3 if their corresponding terms are adjacent in the sequence.  

If (B), note that $\alpha$ in $G$ is not adjacent to the final leaf $\ell$, since if it were then $S=S'$ ($a_\alpha$ is the first point at which they differ).  However $\alpha$ in $G'$ \emph{is} adjacent to $\ell$, meaning $d_{G'}(\alpha,\ell)=3<d_G(\alpha,\ell)$, and so the graphs are not isomorphic.
\end{proof}

\begin{cor}\label{c:number.of.networks}
The number of tier $i$ phylogenetic networks on $X$, with $|X|=\ell$ and $i\ge 1$, is
\[
|\N_i(X)|\ge \frac{(\ell+i-3)!}{(i-1)!}.
\]
\end{cor}
\begin{proof}
The number of distinct echidna graphs with $\pi=id$ is at least the number of sequences in $\Ss(\ell,i-1)$, namely $\frac{(\ell+i-3)!}{(i-1)!}$, and the set of such echidna graphs is a subset of the set of tier $i$ phylogenetic networks.
\end{proof}

Note, this result also holds for $i=0$ because there are $(2\ell-5)!!$ trees and $(2\ell-5)!!\ge (\ell-3)!$.  

\begin{rem}
It would be good to be able to remove the $(i-1)!$ from the denominator of the bound in Corollary~\ref{c:number.of.networks}.  One way to achieve this might be to count echidna networks for general $\pi\in S_{i-1}$, but it seems that this is not trivial.
\end{rem}

\section{A lower bound on the NNI diameter.}\label{s:NNI.lower.bound}

In this section we provide a lower bound on the maximum distance between two tier $i$ phylogenetic networks under NNI operations (Theorem~\ref{t:NNI.lower.bound}).  Our strategy follows that of Li, Tromp and Zhang~\cite{li1996nearest}, who construct bounds for a similar NNI diameter on tree-space.  The strategy involves first bounding the number of networks in a ball of given radius around a network (Proposition~\ref{p:number.reachable}), then using upper and lower bounds on the size of a factorial (Lemma~\ref{l:factorial.bounds}).  For the former of these, we follow~\cite{li1996nearest} in using the concept of a ``graph grammar'', from Sleator, Tarjan and Thurston~\cite{sleator1992short}.

\begin{prop}\label{p:number.reachable}
The number of networks in $\N_i(X)$ reachable in $m$ or fewer NNI operations from any given network in $\N_i(X)$ is at most $6^{2(\ell+i-1)+10m}$.
\end{prop}

\begin{proof}
Define a graph grammar by the three ``productions'' shown in Figure~\ref{f:NNI.productions} (to use the language of Sleator et al~\cite{sleator1992short}).  There is one ``triplet'' production (see Fig.~\ref{f:NNI.productions}(i)) and two ``quartet'' productions (see Figs.~\ref{f:NNI.productions}
(ii) and (iii)) .

\begin{figure}[ht]
\begin{center}
\includegraphics[width=\textwidth]{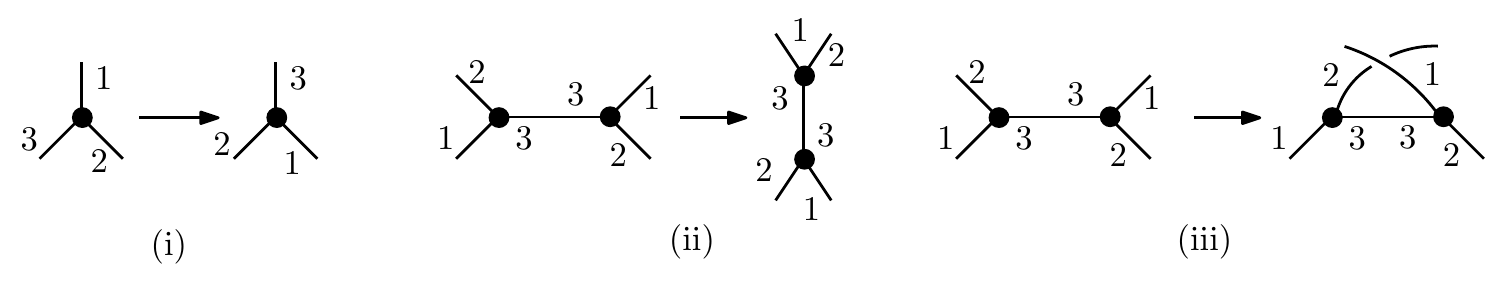}
\end{center}
\caption{The graph grammar of productions that implement NNI operations. The labels are on half edges. }\label{f:NNI.productions}
\end{figure}

For each vertex of degree 3, label half-edges in $N$ arbitrarily by 1,2,3.  
Any NNI operation on a quartet in $N$ involves at most five of the productions in Figure~\ref{f:NNI.productions}: up to two rotations of labels for each vertex, performed by the triplet production, to align the labels with the quartet productions, plus one of the quartet productions.
Thus, a sequence of $m$ NNI operations becomes a sequence of at most $5m$ productions in the graph grammar.  Now, applying Theorem 2.3 of Sleator et al~\cite{sleator1992short}, the number of networks in $\N_i(X)$ reachable in $m$ or fewer steps 
from any network in $\N_i(X)$
is $(c+1)^{n+5rm}$, where $c=5$ (the number of vertices on the left side of the grammar), $r=2$ (the largest number of vertices on the right side of any one production), and $n=2(\ell+i-1)$ (the number of vertices in the network).   This completes the proof.
\end{proof}

Note that the leaf labels in a phylogenetic network are ``tags'' in the sense of~\cite{sleator1992short}, and by \cite[Remark 3.4]{sleator1992short}, this does not change the formula in~\cite[Theorem 2.3]{sleator1992short} for (leaf-labelled) phylogenetic networks.

We will exploit Stirling's well-known formula giving bounds on $m!$, stated below.
 
\begin{lem}[Stirling's formula]\label{l:factorial.bounds}
For $m\ge 1$,
\[
\sqrt{2\pi} \frac{m^{m+\frac{1}{2}}}{e^m}\le m!\le \frac{m^{m+\frac{1}{2}}}{e^{m-1}}.
\]
\end{lem}

\begin{thm}\label{t:NNI.lower.bound}
The diameter $\Delta_i$ of the set of tier $i$ phylogenetic networks, $i\ge 1$, is bounded below by
\[
\Delta_i\ge\frac{1}{20}\left[
(n-3)\log_6\left(\frac{n}{2}-2\right)-(2i-1)\log_6(i-1)-(n-2i)\log_6e-2n
\right].
\]
\end{thm}

\begin{proof}
By Proposition~\ref{p:number.reachable}, the number of networks reachable in $\Delta_i$ operations is at most $6^{n+10\Delta_i}$.  But since this is the diameter, this is all networks.  Thus from Corollary~\ref{c:number.of.networks}, we have 
\begin{equation}\label{e:netwk.no.comparison}
6^{n+10\Delta_i} \ge \frac{(\frac{n}{2}-2)!}{(i-1)!}.
\end{equation}
Using Lemma~\ref{l:factorial.bounds}, with $m=\frac{n}{2}-2$ for the numerator and $m=i-1$ for the denominator of Equation~\eqref{e:netwk.no.comparison}, this gives:
\begin{align*}
6^{n+10\Delta_i} 
&\ge \left[ \sqrt{2\pi} \frac{(\frac{n}{2}-2)^{\frac{n}{2}-\frac{3}{2}}}{e^{\frac{n}{2}-2}} \right]\times
\left[ \frac{e^{i-2}}{(i-1)^{i-\frac{1}{2}}} \right]\\
&= 
 \frac{\sqrt{2\pi}(\frac{n}{2}-2)^{\frac{n}{2}-\frac{3}{2}}}{e^{\frac{n}{2}-i}(i-1)^{i-\frac{1}{2}}}
\end{align*}
Taking logs base 6 and reorganising gives
\begin{align*}
\Delta_i &\ge
\frac{1}{10}\left[
\log_6{\sqrt{2\pi}}+\frac{1}{2}(n-3)\log_6\left(\frac{n}{2}-2\right)-\frac{1}{2}(n-2i)\log_6e-\frac{1}{2}(2i-1)\log_6(i-1)-n
\right]\\
&\ge \frac{1}{20}\left[
(n-3)\log_6\left(\frac{n}{2}-2\right)-(2i-1)\log_6(i-1)-(n-2i)\log_6e-2n
\right],
\end{align*}
as required.
\end{proof}

\section{An upper bound on the NNI diameter}
\label{s:upper.bound.NNI}

In this section we establish an upper bound on the NNI diameter of the space of phylogenetic networks, by providing a schematic NNI path between any two networks.  The maximal length of this path is then an upper bound for the diameter of the space (Theorem~\ref{t:NNI.diam.upper.bd}).

The path we construct is as follows: first convert $N$ into a simple network, and then into a pseudo-Hamiltonian network (defined in Section~\ref{s:prelims}).  Upper bounds for the number of steps in these conversions are given in Lemmas~\ref{l:steps.to.simple.netwk} and~\ref{l:steps.simple.to.pseudo.Ham} respectively.  We then show how to convert any pseudo-Hamiltonian network into any other in a bounded number of steps (Lemma~\ref{l:pseudo.Ham.diam}).  

Finally we remark in Corollary~\ref{c:NNI.dist.btw.arb.netwks} that this result can be used to bound the distance between an arbitrary pair of networks in possibly different tiers.

We begin by deriving an upper bound on the number of non-trivial cut-edges for a network in tier $i$. 

Given a connected graph $G$ with vertex set $V$ and edge set $E$, we define $r(G) = |E|-|V|+1$.
Note that $r(G)$ is clearly the number of edges that need to 
be removed from $G$ in order to obtain a tree that is a spanning tree for $G$.
For a network $N \in \N(X)$, $r(N)$ is known as the {\em reticulation number of $N$}.

\begin{lem}\label{retnum}
	Let $N \in \N(X)$ and $i\geq 0$. Then $N \in \N_i(X)$ if and only if $r(N)=i$.
\end{lem}

\begin{proof}
	We show first that for any network $N\in\N_i(X)$ we have
	$|E(N)|=2\ell-3+3i$. 
	Suppose $N \in \N_i(X)$. Then
	by \cite[Theorem~3]{huber2016transforming}, we can obtain $N$ 
	by first taking some phylogenetic tree (i.e. a network 
	in tier 0) which has $2\ell-3$ edges, 
	then performing $i$ triangle operations (which adds $3i$ edges) to obtain 
	a network in $\N_i(X)$, and then performing some sequence of NNI operations to get $N$
	(which does not change the number of edges). Hence $|E(N)|=2\ell-3+3i$, as required. 
	
	Now, suppose $N \in \N(X)$, some $i\geq 0$. If $N \in \N_i(X)$, 
	then $r(N) = |E(N)|-|V(N)|+1 = (2\ell-3+3i)- 2(\ell+i-1) + 1 = i$.
	Conversely, suppose $r(N)=i$. If $N \in N_j(X)$ some 
	$j\geq 0$, then
	$i=r(N)=|E(N)|-|V(N)|+1 = (2\ell-3+3j) - 2(\ell+j-1) + 1 = j$.
\end{proof}

\begin{prop}\label{l:number.cut-edges}
	Let $N \in \N_i(X)$ some $i\geq 0$, with $n$ vertices.
	The number of non-trivial cut-edges in $N$ is at most $\ell+i-3$.
\end{prop}
\begin{proof}
	Without loss of generality, we may assume that $i\geq 1$ as otherwise $N$ is a phylogenetic tree on $X$ and that the result
	clearly holds.
	We consider the phylogenetic tree $T$ on $X$ that is obtained by shrinking each blob in $N$ 
	down to a vertex. Note that the number of non-trivial cut-edges in $N$ is 
	clearly at most the number of edges in $T$ minus $\ell$.
	
	Now, it follows by \cite[Lemma 6]{huber2016transforming}, that by
	shrinking a blob $B$ of $N$ down to a vertex, we lose 
	at least $r(B)$ vertices.
	But $r(N)$ is the sum of the values $r(B)$  taken over all blobs $B$ in $N$. 
	Hence, since $r(N)=i$ by Lemma~\ref{retnum} and 
	$|V(N)|=2(\ell+i-1)$, the tree $T$ has at most $2\ell+i-2$ vertices, and so
	it has at most $2\ell+i-3$ edges.  The proposition now follows immediately.
\end{proof}

\begin{lem}\label{l:steps.to.simple.netwk}
Suppose $N\in\N_i(X)$ and $i\ge 1$.  We can 
convert $N$ into a simple network by performing at most $\ell+i-3$ NNI operations on $N$.
\end{lem}

\begin{proof}
 Without loss of generality we may assume that $N$
is not simple as otherwise the lemma clearly holds. 
Since $i>0$, $N$ contains at least one blob.
Let $e$ be a non-trivial
cut-edge of $N$ and let $a,b\in V(N)$ such 
that $e=\{a,b\}$. Furthermore, let
$u\in V(N)$ such that $u$ is adjacent with $a$ and let 
$w\in V(N)-\{a\}$ such that
$w$ is adjacent with $b$. Finally, let $C$ denote the connected component of
$N$ containing $a$, obtained by deleting the edge $e$. For the convenience of the
reader, we depict in Figure~\ref{f:NNI.on.cutedge} the case that $C$ contains a 
cycle which shares a vertex with $e$ and that the other vertex of $e$ is not contained 
in a cycle of $N$.
\begin{figure}[ht]
\includegraphics[width=0.5\textwidth]{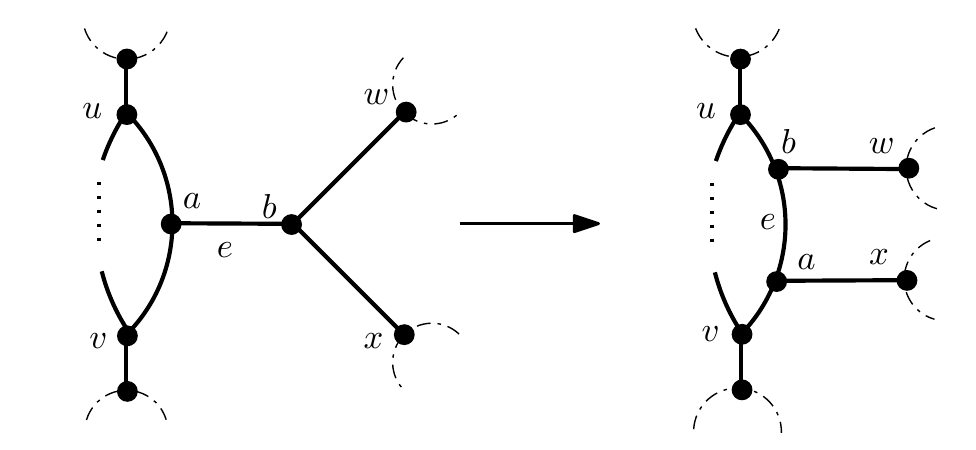}
\caption{Example of an NNI operation on a cut-edge adjacent to a blob.}\label{f:NNI.on.cutedge}
\end{figure}

Since $e$ is a cut-edge of $N$ we can perform an NNI operation on
the path $u,a,b,w$ to obtain a new network $N'$. Since $e$ has been incorporated 
into $C$ in $N'$ and no new cut-edge has been created by that operation, it follows
that $N'$ has one cut-edge less than $N$.
Consequently, by performing at most the number of non-trivial cut-edges
NNI operations, we can convert N into a simple network.
The statement follows, by Proposition~\ref{l:number.cut-edges}.
\end{proof}

We now give an algorithm to convert $N$ into a pseudo-Hamiltonian network.

\begin{lem}\label{l:steps.simple.to.pseudo.Ham}
Suppose $N\in\N_i(X)$ is simple, with $i\ge 1$.  We can convert $N$ into a pseudo-Hamiltonian network by performing at most 
$n$ NNI operations on $N$.
\end{lem}

\begin{proof}
Suppose $N$ is not pseudo-Hamiltonian. Since
$N$ is simple it must contain a blob $B$,
and every cut-edge of $N$ is trivial.
Choose a maximal length cycle $C$ in $B$.  

Choose a non-leaf vertex $v$ that is not in $C$, but is adjacent to a vertex $w\in V(C)$.  
Since $w$ has degree 3 there must exist a vertex $w_1\in V(C)$ that is adjacent with $w$. Note that $w_1$ and $v$ can
not be adjacent, since if they were, they would be in a cycle that is longer
than $C$, violating maximality (the path $w_1$ to $w$ could be extended by going through $v$).  
Since $v$ is not a leaf of $N$, we may choose another vertex $v_1\in V(N)-w$ that is adjacent to $v$. 
Again, since $w$ has degree 3, and is contained in a cycle and adjacent to $v$ (outside the cycle), the
edge $\{w,v_1\}$ also is not contained in $N$. 
Hence, we may apply an NNI operation on the path $w_1,w,v,v_1$ to obtain a new network $N'$.

Note that $N'$ is still a simple network, but the cycle $C$ has been extended to create a new cycle $C'$ with one more vertex, namely $v$.  
To see that $N'$ is still simple, let $v_2\in V(N)- \{w,v_1\}$ be the ``other'' vertex adjacent to $v$ in $N$. Consider the edges $\{v,v_1\}$ and $\{v,v_2\}$ in $N$.  Because $N$ is simple, at most one of $v_1$ and $v_2$ is a leaf, and the other is (or both are) part of a path $w,v,v_i,\dots$ that begins and ends in a vertex in the cycle $C$.  If one of $v_1,v_2$ was a leaf attached to $v$ in $N$, it remains a leaf in $N'$, but is now attached to the cycle $C'$.  For $v_i$ not a leaf, it is part of a path back to $C$ in $N$, and remains part of a path back to $C'$ in $N'$ (a path that is now one vertex shorter).

It follows that a sequence of NNI operations can be performed, each increasing the length of a maximal length cycle in the network by one vertex.  This process ends, since the number of vertices is finite, and it ends with a pseudo-Hamiltonian network.  The number of NNI operations taken is at most the number of vertices in $N$.
\end{proof}

In our final lemma before the main theorem, we bound the distance between any two pseudo-Hamiltonian networks.

\begin{lem}\label{l:pseudo.Ham.diam}
Any two pseudo-Hamiltonian networks in $\N_i(X)$, $i\ge 1$, are at most $\binom{\frac{n}{2}+i-1}{2}$ NNI operations apart, where $n=|V(N)|$.
\end{lem}

\begin{proof}
Fix a pseudo-Hamiltonian cycle for each network.  Both these cycles are of the same length, namely $n-\ell=\ell+2(i-1)$, and both have the same number, $\ell$, of vertices adjacent to leaves. 

The leaves of each network are labelled 1 to $\ell$; number the adjacent vertex of each leaf by the same label.  These vertices are on the pseudo-Hamiltonian cycle.   Now, for each network, and for each non-leaf edge that is \emph{not} on the pseudo-Hamiltonian cycle (each \emph{chord}), number its end vertices by pairs $\{\ell+1,\ell+2\},\dots,\{\ell+2i-3,\ell+2i-2\}$.  This gives every vertex in each network a label ($\ell$ leaves, $\ell$ leaf-adjacent vertices, and $2(i-1)$ vertices contained in chords, for a total of $2(\ell+i-1)$).

For any two adjacent vertices $v_2,v_3$ on a pseudo-Hamiltonian cycle, performing an NNI operation on the length three sub-path $v_1,v_2,v_3,v_4$ has the effect of swapping the middle two adjacent vertices to give the sub-path $v_1,v_3,v_2,v_4$.  Consequently, the arrangement of the vertices labelled $1,\dots,\ell+2(i-1)$ on the pseudo-Hamiltonian cycles can be sorted between the two networks by NNI operations in the number of swaps of adjacent vertices in the cycle.  This is bounded by the diameter of the symmetric group on $\ell+2(i-1)$, which is $\binom{\ell+2(i-1)}{2}=\binom{n-\ell}{2}$, as required, noting that $n-\ell=\frac{n}{2}+i-1$.
\end{proof}

We can now give an upper bound for the diameter of $\N_i(X)$.

\begin{thm}\label{t:NNI.diam.upper.bd}
The diameter $\Delta_i$ of $\N_i(X)$, with $i\ge 1$, is at most $3n+\binom{\frac{n}{2}+i-1}{2}-2$.
\end{thm}
\begin{proof}
By Lemmas 5.2 and 5.3, any network in tier $i\ge 1$ can be transformed into a pseudo-Hamiltonian 
network in at most $\frac{n}{2}-1+n$  steps: $\frac{n}{2}-1=\ell+i-3$ to become simple and $n$ to become pseudo-Hamiltonian. So given any two networks in $\N_i(X)$, they can 
be both made pseudo-Hamiltonian in a total of $3n-2$ steps, and, by Lemma~\ref{l:pseudo.Ham.diam}, 
one can be transformed to the other in $\binom{\frac{n}{2}+i-1}{2}$
steps.
\end{proof}

Note that the same bound for $i=0$ follows immediately from~\cite{li1996nearest} (where the bound is better than in this statement, being essentially $\ell\log \ell$).

We can use Theorem~\ref{t:NNI.diam.upper.bd} to find an upper bound for the distance between \emph{any} pair of networks in $\N(X)$, irrespective of tier.  

\begin{cor}\label{c:NNI.dist.btw.arb.netwks}
Let $N,N'\in \N(X)$, with $N\in\N_i(X)$ and $N'\in\N_j(X)$.  Suppose without loss of generality that $0\leq i\le j$.

Then 
\[d_{\N(X)}(N,N')\le 6\ell+7j-i-8+\binom{\ell+2j-2}{2}.\]
\end{cor}
\begin{proof}
By performing $j-i$ triangle operations starting with $N$ we can create a network $N''$ in tier $j$.
The distance from $N''$ to $N'$ is bounded above by
the diameter bound $\Delta_j$ given in Theorem~\ref{t:NNI.diam.upper.bd} (in which $n$ is the number of vertices in tier $j$, namely $n=2(\ell+j-1)$). 
Hence, the distance from  $N$ to $N'$ is at most $j-i + \Delta_j$, 
and 
\begin{align*}
j-i + \Delta_j
&\le j-i+3n+\binom{\frac{n}{2}+j-1}{2}-2\\
&=j-i+6(\ell+j-1)+\binom{(\ell+j-1)+j-1}{2}-2\\
&=6\ell+7j-i-8+\binom{\ell+2j-2}{2}
\end{align*}
as required.
\end{proof}

\section{SPR and TBR operations}
\label{s:SPR.TBR}

In this section, we  define ``subtree prune and regraft'' (SPR) and ``tree bisection and regraft'' (TBR) operations on network space $\N_i(X)$, $i\ge 1$.  We extend the results of the previous sections on NNI operations to these operations on network space in Section~\ref{s:SPR.TBR.bounds}.  To state these definitions for some network $N\in \mathcal N(X)$, let $v,w\in V(N)$ such that $v$ is of degree 3. Assume that $e=\{v,w\}$ 
is an edge in $N$, but that $\{v_1,v_2\}$ is not an edge in $N$, where $v_1,v_2$ are the vertices other than $w$ incident to $v$ in $N$.

\begin{defn}[SPR operation]\label{d:SPR}
An \emph{SPR operation} on $e$ first removes $e$ from $N$ and then suppresses $v$ (the degree of $v$ is now 2). 
Next, it attaches a new edge $\{w,x\}$ to $w$, where $x$ is a vertex 
subdividing an edge $e'$ of $N$ not incident to $w$.
In case $e$ is a cut-edge of $N$ then we also require that
$e'$ is contained in the connected component not 
containing $w$.
\end{defn}

\begin{defn}[TBR operation]\label{d:TBR}
	Assume that $w$ is such that the degree of $w$ is
	also 3 and that $\{w_1,w_2\}$ is not an edge in $N$
	where $w_1,w_2\in V(N)-\{v\}$ are the other two vertices in $N$ incident with $w$.
  A \emph{TBR operation} on $e=\{v,w\}$ deletes the edge, suppressing the resulting degree 2 vertices $v$ and $w$, and adds a new edge on $N$ between 
  a subdivision vertex of an edge $e_1$ and a subdivision vertex
  of a further edge $e_2$ of $N$. In case $e$ is again a cut-edge of $N$, we also require that $e_1$ and $e_2$ are contained in distinct connected components. 
  
\end{defn}

Note that in Batagelj~\cite{batagelj1981inductive} similar operations are defined on cubic graphs (see 
generating rules P1.-P10).

It is straightforward to check that in network space, each NNI operation
is also an SPR operation, and each SPR operation is also a TBR operation, so we state the following without proof.

\begin{lem}\label{l:NNI.ss.SPR.ss.TBR}
$\NNI\subseteq \SPR\subseteq \TBR$.
\end{lem}

We will write $d_\Theta$ for the distance under the operation $\Theta$, for $\Theta\in\{\NNI,\SPR,\TBR\}$.  Note that these are distances \emph{within} a tier, since each operation $\Theta$ is an operation that remains in a fixed tier.
We have already noted in Section~\ref{s:prelims} that $d_{\NNI}$ is a metric and, in view
of the last lemma, it is straight-forward to check that the same also holds for $d_{\SPR}$ and $d_{\TBR}$.

In fact any TBR operation can be done by just two SPR operations, giving the following relationship among corresponding distances:

\begin{lem}\label{l:SPR.le.2TBR}
The TBR distance $d_{\TBR}(N,N')\le 2d_{\SPR}(N,N')$, the SPR distance, for networks $N,N'\in\N_i(X)$.
\end{lem}
\begin{proof}
Each TBR operation on an edge $e$ of a network $N\in\mathcal N(X)$
can be performed by a pair of SPR operations where in the first SPR
operation the role of $v$ is played by one of the two vertices 
incident with $e$ and, in the second, that role is played by the
other vertex incident with $e$.
\end{proof}

\section{Upper and lower bounds on the SPR and TBR diameters of $\N_i(X)$}
\label{s:SPR.TBR.bounds}

We write $\N_i^\Theta(X)$ for the space of networks in tier $i$ under the operation $\Theta\in\{\NNI,\SPR,\TBR\}$, and write $\Delta_i^\Theta$ for the diameter of $\N_i^\Theta(X)$.

The number of SPR operations from any given network $N$ in tier 
$i\geq 0$ can be given an upper bound as follows.  

First, there is the number of edges one may choose for the operation.  The number of edges is half the total degree, which is $3n-2\ell$ (each vertex has degree 3 except the leaves, which have degree 1).
Note, $3n-2\ell=2(n+i-1)$, and so the number of edges is $n+i-1$.

Each edge $e$ has two end vertices that may be chosen to be detached, and then one may regraft $e$ on to any edge except $e$ itself and the edges still incident to it: $n+i-4$ choices.

Thus there are at most
\[
2(n+i-1)(n+i-4)
\] 
networks reachable from any network in $\N_i(X)$ by applying one SPR operation.

Setting $d=\Delta_i^{\SPR}$, following the previous logic of Section~\ref{s:NNI.lower.bound}, we have that the number of networks in $\N_i(X)$ is at most $(2(n+i-1)(n+i-4))^d$, and so we have 
\[
(2(n+i-1)(n+i-4))^d\ge 
	\frac{\sqrt{2\pi}(\frac{1}{2}n-2)^{\frac{1}{2}n-\frac{3}{2}}}{e^{\frac{1}{2}n-i}(i-1)^{i-\frac{1}{2}}}
\]
using the calculation in the proof of Theorem~\ref{t:NNI.lower.bound}.
Taking natural logs:
\[
d\left[\ln 2+\ln (n+i-1)+\ln (n+i-4)\right]
\ge
\frac{1}{2}(n-3)\ln \left(\frac{n}{2}-2\right)-\frac{1}{2}(2i-1)\ln(i-1)-\frac{1}{2}(n-2i).
\]
Therefore,
\begin{align*}
\Delta_i^{\SPR}&\ge\frac{(n-3)\ln \left(\frac{n}{2}-2\right)-(2i-1)\ln(i-1)-(n-2i)}{2(\ln 2+\ln (n+i-1)+\ln (n+i-4))}\\
&\ge \frac{(n-3)\ln \left(\frac{n}{2}-2\right)-(2i-1)\ln(i-1)-(n-2i)}{4\ln2(n+i)}
\end{align*}

This lower bound on the diameter $\Delta_i^{\SPR}$ gives us one for the TBR diameter, noting Lemma~\ref{l:SPR.le.2TBR}:
\begin{prop}\label{p:lower.bd.SPR.TBR}
\[\Delta_i^{\TBR}\ge\frac{1}{2}\Delta_i^{\SPR}\ge \frac{(n-3)\ln \left(\frac{n}{2}-2\right)-(2i-1)\ln(i-1)-(n-2i)}{8\ln2(n+i)}.
\] 
\end{prop}

To obtain upper bounds for $\Delta_i^{\SPR}$ and $\Delta_i^{\TBR}$, we can similarly follow our approach from the NNI case.

To move from a phylogenetic network $N\in\N_i(X)$ to another network $N'\in\N_i(X)$, first, convert the phylogenetic networks into pseudo-Hamiltonian forms, $N_1$ and $N_1'$.  This takes at most $2n$ moves for each network, since that's how many NNI moves it takes (Lemmas~\ref{l:steps.to.simple.netwk} and~\ref{l:steps.simple.to.pseudo.Ham}).  
Combining Lemmas~\ref{l:pseudo.Ham.diam} and~\ref{l:NNI.ss.SPR.ss.TBR}, at most $n^2$ SPR moves are needed to transform $N_1$ to $N_1'$.
This gives an upper bound for the SPR diameter of
\[\Delta_i^{\SPR}\le n^2+4n.\]

Since each SPR move is also a TBR move (Lemma~\ref{l:NNI.ss.SPR.ss.TBR}), the number of TBR moves between any two networks is at most the maximum number of SPR moved.  That is, $d_{\TBR}\le d_{\SPR}$, which gives an upper bound on $\Delta_i^{\TBR}$.  These upper bounds are summarized as follows:

\begin{prop}\label{p:TBR.SPR.upper.bds}
\[
\Delta_i^{\TBR}\le\Delta_i^{\SPR}\le n^2+4n.
\]
\end{prop}

Both upper bounds in Proposition~\ref{p:TBR.SPR.upper.bds} could be improved by an improvement on the upper bound for the number of SPR moves required to move between two pseudo-Hamiltonian networks.  Whether that bound of $n^2$ can be improved is a question that may be of independent interest.

\section{Discussion}\label{s:discussion}

In this paper, we have presented upper and lower bounds for the diameter of the metric 
$d_{\Theta}$ on $\N^{\Theta}_i(X)$, $\Theta \in\{ \NNI,\SPR,\TBR\}$. 
It would be interesting to know if these bounds can be improved upon
and how close they are to being sharp. We suspect that the lower bound
given in Theorem~\ref{t:NNI.lower.bound} could be improved by 
finding a larger lower bound for the number of networks in $\N_i(X)$
than the one given in Corollary~\ref{c:number.of.networks}, but have not been able to show this.

It would also be of interest to obtain a deeper understanding of the relationship
between the structure of the space $\N(X)$ under $d_{\Theta}$ and the subspace obtained by
restricting $d_{\Theta}$ to the tier $\N_i^{\Theta}(X)$,
$i\geq 0$. 
For example, it is clear that $\N_i^{\Theta}(X)$ is not an isometric subspace of $\N(X)$ under the metric $d_{\Theta}$ for $i \ge 1$, by virtue of the following example.
Take a network $N$ in tier $i$, and use the triangle operation to blow up a vertex $v$, giving a new network $N'$ in tier $i+1$. Now repeat this operation on $N$ but this time use the triangle operation on a different vertex $w\neq v$, to get a different network $N''$ in tier $i+1$.  The distance between $N'$ and $N''$ in $\N(X)$ is 2, through two judicious uses of the triangle operation.  But the distance between them staying within tier $i+1$ is strictly greater than 2, regardless of which operation of NNI, SPR or TBR is used.

In this paper we have considered unrooted networks. 
However, it would be very interesting to see how our
results could be extended to rooted networks. 
Some results concerning spaces of rooted networks are presented in \cite{radice2012} 
and \cite{yu2014}.  However, it is still necessary to define operation-based metrics 
on these spaces, and previous work on spaces of level-1 rooted networks \cite{huber2012}
suggests that this could be quite technical.
Moreover, to find diameter bounds on the resulting space of rooted network metrics such
as the one given in Theorem~\ref{t:NNI.lower.bound}, it may be 
necessary either to introduce a new approach for dealing with graph grammars
arising from directed graphs (which are not considered in \cite{sleator1992short}), or to avoid this 
method of proof completely.

There are also some interesting computational questions concerning the 
metrics $d_{\Theta}$. For example, what is the computational complexity 
of computing $d_{\Theta}$? Note that the NNI, SPR and TBR distance are
all NP-complete to compute (cf. \cite{dasgupta1997distances,hickey2008,allen2001subtree}). 
In light of this fact, it is likely that 
the metric $d_{\Theta}$ is also NP-complete to compute. One way to show this
could be to prove that $\N_0^{\Theta}(X)$ (i.e. tree-space) is an 
isometric subgraph of $\N^{\Theta}(X)$ under $d_{\Theta}$, which is 
a special case of the problem mentioned above. 

In this paper we have considered discrete spaces
of networks. However, it would be interesting to define 
and study continuous variants of these spaces. Continuous tree-spaces have been defined
and studied \cite{billera2001geometry}, and arise since real-valued edge-lengths are often assigned to phylogenetic trees.
How should we formally define continuous spaces of networks with edge-weights and metrics on these spaces, and 
what are their properties? Note that recently a definition for 
a continuous space of unrooted networks has been proposed \cite{devadoss2016},
and shown to have interesting geometric properties.\\

\noindent{\bf Acknowledgments.}
KTH and VM thank the London Mathematical Society for its support and also 
the Centre for Research in Mathematics at Western Sydney University, Australia, for hosting them during this research.

\bibliographystyle{plain}

\end{document}